\documentclass[10pt,conference]{IEEEtran}
\IEEEoverridecommandlockouts
\makeatletter

\def\endthebibliography{%
\def\@noitemerr{\@latex@warning{Empty `thebibliography' environment}}%
\endlist
}
\makeatother

\usepackage{cite}
\usepackage{amsmath,amssymb,amsfonts}
\usepackage[ruled,vlined]{algorithm2e}    
\usepackage{setspace}
\usepackage{algorithmic}

\usepackage{graphicx}
\usepackage{textcomp}
\usepackage{amsthm}
\def\BibTeX{{\rm B\kern-.05em{\sc i\kern-.025em b}\kern-.08em
	T\kern-.1667em\lower.7ex\hbox{E}\kern-.125emX}}
	
\usepackage[left=0.625in,right=0.625in,top=0.72in,bottom=0.99in]{geometry}

\usepackage{setspace}
\usepackage{multirow}
\usepackage{comment}
\usepackage{tabularx}
\usepackage{booktabs}
\usepackage{caption}
\usepackage{graphicx}
\usepackage{subcaption}
\usepackage[table,xcdraw]{xcolor}

\newtheorem{prop}{Proposition}

\usepackage{xspace} 
\usepackage[acronym,nogroupskip,nonumberlist,nopostdot]{glossaries}

\newcommand{\vectr}[1]{\boldsymbol{\mathrm{#1}}}
\newcommand{\matr}[1]{\boldsymbol{\mathrm{#1}}}

\newcommand\norm[1]{\left\lVert#1\right\rVert}

\newcommand{\eye}{\matr{I}}
\newcommand{\pilot}{\matr{\Phi}}
\newcommand{\RomanNumeralCaps}[1]{\MakeUppercase{\romannumeral #1}}

\newcommand{\be}{\vectr{e}}

\newcommand{\bc}{\vectr{c}}
\newcommand{\bg}{\vectr{g}}
\newcommand{\medium}{m}

\newcommand{\bW}{\matr{W}}

\newcommand{\ba}{\vectr{a}}

\newcommand{\bE}{\matr{E}}

\newcommand{\bw}{\vectr{w}}

\newcommand{\bX}{\matr{X}}
\newcommand{\bx}{\vectr{x}}

\newcommand{\herm}{\mathsf{H}}
\newcommand{\transp}{\mathsf{T}}
\newcommand{\ap}{\text{AP}}
\newcommand{\trp}{\mathsf{T}}

\newcommand{\realset}[2]{ \mathbb{R}^{#1 \times #2}  }
\newcommand{\complexset}[2]{ \mathbb{C}^{#1 \times #2}  }

\newacronym{4g}{4G}{fourth-generation}
\newacronym{5g}{5G}{fifth-generation}
\newacronym{6g}{6G}{sixth-generation}
\newacronym{2d}{2D}{two-dimensional}
\newacronym{3d}{3D}{three-dimensional}
\newacronym{5gnr}{5G NR}{5G New Radio}
\newacronym{3gpp}{3GPP}{third-generation partnership project}
\newacronym{adc}{ADC}{analog-to-digital converter}
\newacronym{am}{AM}{amplitude modulation}
\newacronym{ambc}{AmBC}{ambient BC}
\newacronym{ap}{AP}{access point}
\newacronym{ar}{AR}{augmented reality}
\newacronym{aoa}{AOA}{angle-of-arrival}
\newacronym{agc}{AGC}{automatic gain control}
\newacronym{awgn}{AWGN}{additive white Gaussian noise}
\newacronym{bc}{BC}{backscatter communication}
\newacronym{bde}{BD}{backscatter device}
\newacronym{bf}{BF}{beamforming}
\newacronym{ber}{BER}{bit error rate}
\newacronym{bs}{BS}{base station}
\newacronym{bibc}{BiBC}{bistatic BC}
\newacronym{ce}{CE}{carrier emitter}
\newacronym{cmimo}{C-MIMO}{collocated MIMO}
\newacronym{cdf}{CDF}{cumulative distribution function}
\newacronym{csi}{CSI}{channel state information}
\newacronym{cvx}{CVX}{convex optimization toolbox}
\newacronym{dac}{DAC}{digital-to-analog converter}
\newacronym{dbf}{DPBF}{dual-polarization BF}
\newacronym{dmimo}{D-MIMO}{distributed MIMO}
\newacronym{dl}{DL}{downlink}
\newacronym{doa}{DOA}{direction-of-arrival}
\newacronym{dli}{DLI}{direct link interference}
\newacronym{dft}{DFT}{discrete Fourier transform}
\newacronym{dtft}{DTFT}{discrete-time Fourier transform}
\newacronym{dp}{DP}{dynamic programming}
\newacronym{en}{EN}{energy neutral}
\newacronym{end}{END}{energy neutral device}
\newacronym{eirp}{EIRP}{effective isotropic radiated power}
\newacronym{etsi}{ETSI}{European Telecommunications Standards Institute}
\newacronym{evd}{EVD}{eigenvalue decomposition}
\newacronym{fdd}{FDD}{frequency-division duplexing}
\newacronym{fdma}{FDMA}{frequency-division multiple access}
\newacronym{fft}{FFT}{fast Fourier transform}
\newacronym{gs}{GS}{grid search}
\newacronym{gd}{GD}{gradient descent}
\newacronym{gsm}{GSM}{Global System for Mobile Communications}  
\newacronym{gna}{GNA}{Girvan-Newman algorithm}
\newacronym{glrt}{GLRT}{generalized log-likelihood ratio test}
\newacronym{icsi}{ICSI}{imperfect CSI}
\newacronym{iot}{IoT}{Internet-of-Things}
\newacronym{iid}{i.i.d.}{independent and identically distributed}
\newacronym{isr}{ISR}{interference-to-signal ratio}
\newacronym{ieee}{IEEE}{Institute of Electrical and Electronics Engineers}
\newacronym{los}{LoS}{line-of-sight}
\newacronym{lora}{LoRa}{long range}
\newacronym{lti}{LTI}{linear time-invariant}
\newacronym{ls}{LS}{least-squares}
\newacronym{lte}{LTE}{Long-Term Evolution}
\newacronym{lan}{LAN}{local area network}
\newacronym{lsb}{LSB}{least significant bit}
\newacronym{m}{m}{meters}
\newacronym{ml}{ML}{maximum-likelihood}
\newacronym{mse}{MSE}{mean square error}
\newacronym{mimo}{MIMO}{multiple-input multiple-output}
\newacronym{mumimo}{MU-MIMO}{multi-user \gls{mimo}}
\newacronym{miso}{MISO}{multiple-input single-output}
\newacronym{mmwave}{mmWave}{millimeter wave}
\newacronym{mmse}{MMSE}{minimum mean square error}
\newacronym{map}{MAP}{maximum a posteriori probability}
\newacronym{mrc}{MRC}{maximum-ratio combining}
\newacronym{mrt}{MRT}{maximum-ratio transmission}
\newacronym{mobc}{MoBC}{monostatic BC}
\newacronym{nb}{NB}{narrowband}
\newacronym{nmse}{NMSE}{normalized mean square error}
\newacronym{nr}{NR}{New Radio}
\newacronym{np}{NP}{Neyman-Pearson}
\newacronym{nfc}{NFC}{near-field communication}
\newacronym{nlos}{NLoS}{non-line-of-sight}
\newacronym{ofdm}{OFDM}{orthogonal frequency division multiplexing}
\newacronym{ofdma}{OFDMA}{orthogonal frequency-division multiple access}
\newacronym{ostbc}{OSTBC}{orthogonal space–time block code}
\newacronym{ota}{OtA}{over-the-air}
\newacronym{ua}{UA}{uncertainty agnostic}
\newacronym{p1}{P1}{Phase \RomanNumeralCaps{1}}
\newacronym{p2}{P2}{Phase \RomanNumeralCaps{2}}
\newacronym{pl}{PL}{path loss}
\newacronym{pana}{PanA}{Panel A}
\newacronym{panb}{PanB}{Panel B}
\newacronym{pgd}{PGD}{projected gradient descent}
\newacronym{ple}{PLE}{path loss exponent}
\newacronym{pcsi}{PCSI}{perfect \gls{csi}}
\newacronym{papr}{PAPR}{peak-to-average power ratio}
\newacronym{pg}{PG}{path gain}
\newacronym{pdf}{pdf}{probability density function}
\newacronym{phy}{PHY}{physical layer}
\newacronym{psd}{PSD}{positive semidefinite}
\newacronym{rcs}{RCS}{radar cross section}
\newacronym{Riss}{RIS}{Reconfigurable intelligent surfaces}
\newacronym{ris}{RIS}{reconfigurable intelligent surface}
\newacronym{riss}{RIS}{reconfigurable intelligent surfaces}
\newacronym{rf}{RF}{radio frequency}
\newacronym{rfid}{RFID}{radio frequency identification}
\newacronym{rms}{RMS}{root mean square}
\newacronym{rss}{RSS}{received signal strength}
\newacronym{rv}{RV}{random variable}
\newacronym{sdma}{SDMA}{space-division multiple access}
\newacronym{sdr}{SDR}{semidefinite relaxation}
\newacronym{si}{SI}{self-interference}
\newacronym{sic}{SIC}{successive interference cancellation}
\newacronym{sumimo}{SU-MIMO}{single-user \gls{mimo}}
\newacronym{svd}{SVD}{singular value decomposition}
\newacronym{smc}{SMC}{specular multipath component}
\newacronym{snr}{SNR}{signal-to-noise ratio}
\newacronym{sinr}{SINR}{signal-to-interference-plus-noise ratio}
\newacronym{sir}{SIR}{signal-to-interference ratio}
\newacronym{siso}{SISO}{single-input single-output}
\newacronym{simo}{SIMO}{single-input multiple-output}
\newacronym{tdoa}{TDOA}{time-difference-of-arrival}
\newacronym{toa}{TOA}{time-of-arrival}
\newacronym{tdd}{TDD}{time division multiplexing}
\newacronym{tdma}{TDMA}{time-division multiple access}
\newacronym{ue}{UE}{user equipment}
\newacronym{ul}{UL}{uplink}
\newacronym{uhf}{UHF}{ultra high frequency}
\newacronym{ula}{ULA}{uniform linear array}
\newacronym{upa}{UPA}{uniform planar array}
\newacronym{ura}{URA}{uniform rectangular array}
\newacronym{uwb}{UWB}{ultrawideband}
\newacronym{zf}{ZF}{zero-forcing}
\newacronym{qam}{QAM}{quadrature amplitude modulation}
\newacronym{qos}{QoS}{Quality of Service}
\newacronym{wlan}{WLAN}{wireless local area network}
\newacronym{wpt}{WPT}{wireless power transfer}
\newacronym{wrt}{w.r.t.}{with respect to}
\newacronym{wb}{WB}{wideband}

\begin{document}
\bstctlcite{IEEEexample:BSTcontrol}

\title{Analysis of Broad Beam Beamforming for Collocated and Distributed MIMO}

\author{\IEEEauthorblockN{Ahmet Kaplan, Diana P. M. Osorio, and Erik G. Larsson}
	\IEEEauthorblockA{Department of Electrical Engineering (lSY), Linköping University, 581 83 Linköping, Sweden.}
\vspace{-20pt}
}


\maketitle
\vspace{-5pt}
\begin{abstract}
Broad beam \gls{bf} design in \gls{mimo} can be convenient for initial access, synchronization, and sensing capabilities in cellular networks by avoiding overheads of sweeping methods while making efficient use of resources.
Phase-only \gls{bf} is key for maximizing power efficiency across antennas.
A successful method to produce broad beams is the phase-only \gls{dbf}. However, its efficiency has not been proved in \gls{nlos}. Therefore, this paper contributes by evaluating \gls{dbf}  in collocated and distributed \gls{mimo} configurations under both \gls{los} and \gls{nlos} channel conditions.
We model the reflection coefficients for different materials in \gls{nlos} conditions and propose the use of orthogonal space-time block code to improve the coverage compared to the \gls{dbf} in \gls{cmimo}.
We further propose a \gls{dbf} method for distributed \gls{mimo} and show that it achieves better coverage than \gls{cmimo} with \gls{dbf}.  
\end{abstract}

\begin{IEEEkeywords}
Beamforming, broad beam, distributed \gls{mimo}, dual-polarized antennas, reflection coefficients. \vspace{-5pt}
\end{IEEEkeywords}
\glsresetall

\section{Introduction}

\Gls{mimo} is a successful technology to enhance data rates, signal quality, coverage, and energy efficiency through \gls{bf}. As the number of antennas in an \gls{ap} increases, user-specific information can be transmitted using narrower beams, improving \gls{snr}. However, in \gls{4g} and \gls{5g} networks, broadcast signals for synchronization and initial access require wide-area coverage.
When integrating sensing into the communication networks, the wide coverage can also be used for target search and detection.

The most common methods to broadcast information are beam sweeping \cite{3gpp_ts_38_802_2017,raghavan2016beamforming} and space-time block codes \cite{li2021construction, meng2016omnidirectional,karlsson2018techniques}, but they require multiple time or frequency slots.
In contrast, a broad beam generated by a multi-antenna system using a single time and frequency slot is a more efficient solution. 

Several \gls{bf} methods have been proposed to generate a broad beam \cite{qiao2016broadbeam,zhang2017energy,sergeev2017enhanced,leifer2016revisiting} by 
adjusting the amplitude and phase of the \gls{bf} coefficients. 
However, varying amplitudes can reduce energy efficiency, as practical systems often use separate power amplifiers for each antenna. To fully utilize the available power, phase-only \gls{bf}, where the coefficients have the same amplitude but differ in phase, is preferred \cite{petersson2020power}.
However, single-polarized antennas using phase-only \gls{bf} cannot achieve broad beams that mimics the radiation pattern of a single antenna using a single time and frequency slot. 

On the other hand, dual-polarized antennas can enable a broad beam using phase-only \gls{bf}. The \gls{dbf} technique, which mimics the single antenna radiation using multiple antennas in \gls{los} channel, was introduced in \cite{petersson2020power, girnyk2020simple}. In \cite{girnyk2020simple}, \gls{bf} coefficients are proposed for two cross-polarized antennas, while \cite{petersson2022energy} and \cite{girnyk2021efficient} extend the \gls{dbf} technique to flexible antenna array sizes. 
In \cite{simonsson2021dual}, \gls{dbf} method is experimentally demonstrated.
The \gls{dbf} method is also extended to \gls{ris} scenarios in \cite{ramezani2024broad,ramezani2023dual} to reflect signals in a broad angle with a dual-polarized \gls{ris}.

However, there are two significant gaps in the existing
literature. First, the performance of the \gls{dbf} technique has not
been thoroughly investigated under \gls{nlos} channel. All studies focus on evaluating \gls{dbf} under \gls{los}, except for~\cite[Fig. 9]{petersson2022energy} that investigates the received power distribution in the urban macro scenario and \cite{simonsson2021dual} that investigates \gls{dbf} by introducing practical experiments. The impact of multipath and attenuation in \gls{nlos} channels remains unexplored, and requires further study to understand the robustness and performance of \gls{dbf}.

Second, there is a lack of research addressing the extension of \gls{dbf} designs for \gls{dmimo} systems. \gls{dmimo} setups, where antennas are spread across a large area, provide substantial benefits in terms of coverage and macro-diversity. Developing a \gls{dbf} strategy for \gls{dmimo} is critical to exploit the advantages of macro-diversity and ensure efficient signal delivery across a wide area.

Our contribution can be summarized as follows:
\begin{itemize}    
    \item Using available reflection coefficient models, we construct a channel model for \gls{nlos} scenario that considers material-dependent reflections, polarization effects, and the directions of the received electric fields.
    \item We analyze the coverage of the \gls{cmimo} with \gls{dbf} and compare the results to those of \gls{cmimo} with \gls{ostbc} both for \gls{nb} and \gls{wb} scenarios. 
    \item We propose a \gls{bf} method to broadcast information in a \gls{dmimo} setup, aiming to enhance the coverage. 
    \item We compare the performances of \gls{cmimo} and \gls{dmimo}, and examine how the macro-diversity provided by distributed antennas impacts the coverage.
\end{itemize}

\begin{figure*}[tbp]
    \centering
    \begin{minipage}{0.31\textwidth}
        \centering
        \includegraphics[width=\linewidth]{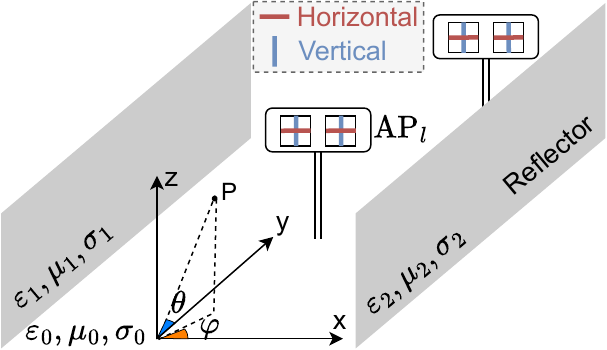}
        \caption{System model.}
        \label{fig:System_Model}
    \end{minipage}
    \hfill
    \begin{minipage}{0.30\textwidth}
        \centering
        \includegraphics[width=\linewidth]{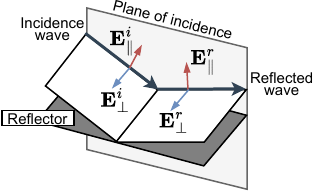}
        \caption{Electric field directions of a uniform plane wave.}
        \label{fig:e_field}
    \end{minipage}    
    \hfill
    \begin{minipage}{0.24\textwidth}
        \centering
        \includegraphics[width=\linewidth]{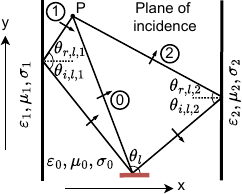}
        \caption{Electric field direction for horizontal polarization.}
        \label{fig:e_field_direc}
    \end{minipage}
    \vspace{-15pt}
\end{figure*}

\textbf{Notation:} 
$(\cdot)^\trp$, $(\cdot)^\herm$, and $(\cdot)^*$ represent the transpose, Hermitian transpose, and conjugate, respectively. 
Bold letters denote matrices (uppercase) and vectors (lowercase), whereas italic letters are used for scalars.
The Euclidean norm is $\norm{\bx}$. The notation $[\bX]_{i,j}$ indicates the element in the $i$-th row and $j$-th column of the matrix $\bX$. The $i$-th element of the vector $\bx$ is denoted by $x_i$. All vectors are assumed to be column vectors.
$\eye_M$ denotes the $M \times M$ identity matrix. Lastly, $\mathbb{C}$ and $\mathbb{R}$ represent the fields of complex and real numbers, respectively.

\section{Signal Propagation Framework}
We consider $L$ \glspl{ap}, and $\ap_l$ has $K_l$ dual-polarized antennas, where $l \in \{1,\dotsc,L\}$. 
Each dual-polarized antenna
has two elements with identical radiation patterns \cite{girnyk2021efficient} and orthogonal polarizations. The system model is given in Fig.~\ref{fig:System_Model}.\footnote{Note that in Fig. \ref{fig:System_Model}, the antennas are sketched as lines to show the polarization, and they are not dipole antennas.}

\subsection{Channel Model}
We model the wireless channel for both parallel ($\parallel$) and perpendicular ($\perp$) polarizations. The subscripts $\parallel$ and $\perp$ denote the components of the electric field aligned parallel and perpendicular to the plane of incidence, respectively. The direction of the electric field is given in Fig. \ref{fig:e_field}, where $\bE^i_\chi$ and $\bE^r_\chi$ represent the incoming and the reflected electric fields for polarization $\chi \in \{\parallel, \perp\}$ \footnote{Note that the transmitted field is not given in the figure, as it is not relevant.}.
For a given polarization $\chi$, the channel between the $k_l$-th antenna of $\ap_l$ and a receiver is modeled as follows
\vspace{-5pt}
\begin{equation} \label{eq:channel}
    g_{\chi, l, k_l} = \frac{\lambda}{4 \pi d_{l, k_l}} e^{-j \frac{2\pi}{\lambda} d_{l, k_l}} + \sum_{\medium=1}^{M} \frac{\gamma_{\chi, l,\medium} \lambda}{4 \pi d_{l, k_l,\medium}} e^{-j \frac{2\pi}{\lambda} d_{l, k_l,\medium}},
\end{equation}
where $k_l \in \{1, \dotsc, K_l\}$, $\lambda$ is the wavelength of the transmitted signal, and $d_{l,k_l}$ represents the direct distance from the $k_l$-th antenna of $\ap_l$ to the receiver.  
The channel model in Eq.~\eqref{eq:channel} comprises two distinct components: the first term corresponds to the \gls{los} component, while the summation term accounts for the reflections from the surrounding environment.

The wave is radiated in
free space with $M$ infinite planar reflectors contributing to
the \gls{nlos} paths and indexed as $\medium \in \{1,
\dotsc, M\}$. 
The parameter $d_{l, k_l,\medium}$ is the
\gls{nlos} path-length between the $k_l$-th antenna of $\ap_l$ and the
receiver due to the first-order reflection from the reflector $m$.  Additionally, the term $\gamma_{\chi,l,\medium} \in \mathbb{C}$ represents the reflection
coefficient associated with the reflector $\medium$ for an incoming signal
from $\ap_l$ in the $\chi$ polarization. This coefficient represents
the amplitude and phase changes induced by the reflection.

\subsection{Modeling of Reflection Coefficient}
We assume that the reflectors are in the far field of each \gls{ap},
and the incoming waves to the reflectors are planar.
Thus, we model a single reflection coefficient per \gls{ap} rather than per antenna.
The reflection coefficient $\gamma_{\chi, l,\medium}$ is given by~\cite{balanis2012}
\footnote{For parallel polarization, we use the reverse of the electric field direction provided in \cite[Fig. 5.4]{balanis2012}. Consequently, $\gamma_{\parallel,l,\medium}$ in \cite{balanis2012} is multiplied by 
$-1$.}
\vspace{-4pt}
\begin{subequations} \label{eq:reflection_coeff}
\begin{align}
    \gamma_{\perp,l,\medium} &= \frac{\eta_{\medium} \cos\theta_{i,l,\medium} - \eta_0 \cos\theta_{t,l,\medium}}{\eta_{\medium} \cos\theta_{i,l,\medium} + \eta_0 \cos\theta_{t,l,\medium}}, \\
    \gamma_{\parallel,l,\medium} &= \frac{\eta_0 \cos\theta_{i,l,\medium} -\eta_{\medium} \cos\theta_{t,l,\medium}}{\eta_0 \cos\theta_{i,l,\medium} + \eta_{\medium} \cos\theta_{t,l,\medium}},
\end{align}
\end{subequations}
\vspace{-5pt}
where
\begin{itemize}
    \item $\theta_{i,l,\medium} \in [0,\pi/2]$ and $\theta_{t,l,\medium} \in \mathbb{C}$ are the angles of incidence and refraction, respectively, between the $\medium$-th reflector and incoming wave from $\ap_l$,
    \item $\eta_0 \in \mathbb{C}$ is the intrinsic impedance of free space,
    \item $\eta_{\medium} \in \mathbb{C}$ is the intrinsic impedance of the $\medium$-th reflector.
\end{itemize}
Note that $\theta_{t,l,\medium}$ is complex and not the true angle of refraction if the $m$-th reflector is a lossy conductor.
The intrinsic impedance is defined as  \cite{balanis2012}
\vspace{-6pt}
\begin{equation}
    \eta_{n} = \sqrt{j\omega \mu_n / (\sigma_n + j\omega \varepsilon_n)},
    \vspace{-2pt}
\end{equation}
where $n \in \{0,\dotsc,M\}$, $\omega = 2 \pi f$ represents the angular frequency, and $\varepsilon_n$, $\mu_n$, and $\sigma_n$ are the permittivity, permeability, and conductance, respectively, for free space (when $n=0$) and $n$-th reflector (when $n>0$).
The relation between $\theta_{i,l,\medium}$ and $\theta_{t,l,\medium}$ is given by Snell's law of refraction as  \cite{balanis2012}
\vspace{-2pt}
\begin{equation}
    \cos(\theta_{t,l,\medium}) = \sqrt{1 - \left( j \beta_0/(\alpha_\medium + j \beta_\medium) \right)^2 \sin^2(\theta_{i,l,\medium})}.
    \vspace{-2pt}
\end{equation}
The attenuation and phase constants are $\alpha_m \! \approx \! \frac{\sigma_m}{2} \sqrt{\mu_m / \varepsilon_m}$ and $\beta_m \!\approx\! \omega \sqrt{\mu_m \varepsilon_m}$, respectively, for good dielectrics, and $\alpha_m \!\approx\! \sqrt{\omega \mu_m \sigma_m/2}$ and $\beta_m \!\approx\! \sqrt{\omega \mu_m \sigma_m/2}$ for good conductors. The phase constant in free space is $\beta_0 \!=\! \omega \sqrt{\mu_0 \varepsilon_0}$ \cite{balanis2012}.

If the wave reflects from a good conductor, $\theta_{t,l,\medium} \approx 0$ 
and $\eta_0/\eta_m \gg 1$. As a result $\gamma_{\perp,l,\medium} \approx -1$ and $\gamma_{\parallel,l,\medium} \approx 1$. 

\subsection{Modeling of Electric Field Direction}\label{sec:e-field_direction}
\vspace{-3pt}
We model the electric field direction for the given system model with two reflectors in Fig. \ref{fig:System_Model}. Without loss of generality, we assume that the \glspl{ap} are equipped with vertically and horizontally polarized antennas. \footnote{The polarization of an antenna is defined with respect to the Earth's surface.} To model the electric field direction, we make the following assumptions: (1) the \glspl{ap} and the receiver are located at the same height, and (2) reflectors are smooth, flat surfaces positioned in the $y$-$z$ plane. Thus, the vertically and horizontally polarized waves correspond to perpendicular and parallel polarized waves, respectively, and the polarization of the reflected wave remains unchanged. In addition, (3) the reflectors are in the far field of the \glspl{ap}.

\textbf{Vertical Polarization:} When vertically polarized antennas of $\ap_l$ transmit, the received waves are assumed to have their electric field aligned with the $z$-axis, i.e., $\vec{u}_{l,n}^\perp = \vec{z}$. 
For polarization $\chi$, $\vec{u}_{l,n}^\chi \in \realset{3}{1}$ is a unit vector in Cartesian coordinates, and shows the direction of the received electric field coming from the \gls{los} path when $n=0$ and from the $n$-th reflector for $n\neq0$.

\textbf{Horizontal Polarization:} When horizontally polarized antennas of $\ap_l$ transmit, the electric field directions are shown in the $x$-$y$ plane
in Fig. \ref{fig:e_field_direc}. The arrows in the figure show the electric field directions, and at point P:
\textcircled{0} $\vec{u}_{l,0}^{\parallel} \!=\! -\vec{y} \cos(\theta_l \!\pmod{\pi}) \!+\! \vec{x} \sin(\theta_l \!\pmod{\pi})$,
\textcircled{1} $\vec{u}_{l,1}^{\parallel} \!=\! -\vec{y} \cos(\theta_{r,l,1}) \!+\! \vec{x} \sin(\theta_{r,l,1})$, and
\textcircled{2} $\vec{u}_{l,2}^{\parallel} \!=\! \vec{y} \cos(\theta_{r,l,2}) \!+\! \vec{x} \sin(\theta_{r,l,2})$. Here, $\theta_{r,l,1}$ and $\theta_{r,l,2}$ are the reflection angles, with $\theta_{r,l,1} = \theta_{i,l,1}$ and $\theta_{r,l,2} = \theta_{i,l,2}$. 
Note that the direction vectors are approximations, as they are calculated per \gls{ap} rather than per antenna. However, this approximation is generally accurate, as the angular variation across antennas within an \gls{ap} is negligible, especially in the far-field region.

Finally, the received signal and \gls{pg} can be calculated using the channel in Eq. \eqref{eq:channel} and the electric field directions in this subsection. The examples for different cases are given in Section \ref{sec:numerical_results}.

\section{Revisiting the Proposed \gls{dbf} for C-MIMO}
\vspace{-2pt}
For \gls{cmimo}, the design of \gls{bf} coefficients to generate a broad beam is presented in \cite{petersson2020power, girnyk2020simple, petersson2022energy, girnyk2021efficient}.
The \gls{bf} coefficient for $\ap_l$ is defined as $\bw_{\text{V},l} \in \complexset{K_l}{1}$ and $\bw_{\text{H},l} \in \complexset{K_l}{1}$, where V and H stand for the vertically and horizontally polarized antennas, respectively. Note that these \gls{bf} coefficients can be applied to any two orthogonal polarizations, such as cross-polarized antennas. 

The radiated electric field from $\ap_l$ is given by
\vspace{-2pt}
\begin{equation}
\be_l(\varphi, \theta) =
    \begin{bmatrix}
        \ba^\trp(\varphi, \theta) \bw_{\text{V},l}  \\
        \ba^\trp(\varphi, \theta) \bw_{\text{H},l} 
    \end{bmatrix} \sqrt{G(\varphi, \theta)},
    \vspace{-2pt}
\end{equation}
where $G(\varphi, \theta)$ is the radiation pattern of a single antenna for both polarizations. For an antenna array deployed along the $y$ axis, the steering vector $\ba(\varphi, \theta)$ is defined by
\begin{equation}
    \mathbf{a}(\varphi, \theta) = 
    \begin{bmatrix}
        1 \\
        \exp(-j \frac{2\pi}{\lambda} d \sin(\varphi) \cos(\theta)) \\
        \vdots \\
        \exp(-j \frac{2\pi}{\lambda} (K_l-1) d \sin(\varphi)\cos(\theta))
    \end{bmatrix},
\end{equation}
where $\varphi \in [-\pi, \pi)$ is azimuth angle and $\theta \in [-\pi/2, \pi/2]$ is elevation angle in spherical coordinates as shown in Fig. \ref{fig:System_Model}, and $d$ is the inter-antennas distance.
As shown in \cite{petersson2020power, girnyk2020simple, petersson2022energy, girnyk2021efficient}, the total radiated power 
pattern of $\ap_l$ is
\vspace{-2pt}
\begin{equation} \label{eq:radiation_pattern}
\begin{split}
    \norm{\be_l(\varphi, \theta)}^2  \!=\! \sum_\chi \!
    \norm{\ba^\trp(\varphi, \theta) \bw_{\chi,l}}^2 \! G(\varphi, \theta) \!=\! 2 K_l G(\varphi, \theta).
    \vspace{-5pt}
\end{split}
\end{equation}

Thus, the radiation pattern of the dual-polarized array matches that of a single antenna element, while the array’s radiated power is $2K_l$ times stronger due to all antennas transmitting at maximum power.

To illustrate this, consider a simple case with two dual-polarized antennas and $G(\varphi, \theta) = 1$. Since there is a single \gls{ap} in this case, the subscript 
$l$ is omitted for notation simplicity. In \cite{girnyk2020simple}, the proposed \gls{bf} coefficients for two dual-polarized antennas are 
\vspace{-3pt}
\begin{equation}
    \bw_\text{V}=[w_{V,1}, w_{V,2}]^\transp,\
    \bw_\text{H}=[w_{H,1}, w_{H,2}]^\transp,
    \vspace{-3pt}
\end{equation}
where $ w_{V,2} = -w_{H,1}^*, w_{H,2} = w_{V,1}^*$ and  each \gls{bf} coefficient has unit magnitude.
The radiated electric field for the two dual-polarized antennas is given by
\vspace{-2pt}
\begin{equation}
		\vectr{e}(\varphi, \theta) = 
		\underbrace{\begin{bmatrix}
			w_{V,1}a_1(\varphi, \theta)\\
			w_{H,1}a_1(\varphi, \theta)
		\end{bmatrix}}_{\vectr{e}_1(\varphi, \theta)}
		+
		\underbrace{\begin{bmatrix}
			w_{V,2}a_2(\varphi, \theta)\\
			w_{H,2}a_2(\varphi, \theta)
		\end{bmatrix}}_{\vectr{e}_2(\varphi, \theta)},
        \vspace{-5pt}
\end{equation}
where $a_1(\varphi, \theta)$ and $a_2(\varphi, \theta)$ are the elements of the steering vector corresponding to the first and second dual-polarized antennas.
The total radiated power pattern is computed as
\begin{equation} \label{eq:total_rad_patrn}
	\begin{split}
		&\norm{\vectr{e}(\varphi, \theta)}^2  \\
		 &\ \ =\! \norm{\vectr{e}_1(\varphi, \theta)}^2 \!+\! \norm{\vectr{e}_2(\varphi, \theta)}^2 \!+\! 2\!\text{ Re}(\vectr{e}_1^\herm(\varphi, \theta) \vectr{e}_2(\varphi, \theta)),
	\end{split}
\end{equation}
where $\norm{\vectr{e}_1(\varphi, \theta)}^2=\norm{\vectr{e}_2(\varphi, \theta)}^2=2$.
The \gls{bf} coefficients in \cite{girnyk2020simple} ensures that
$
\vectr{e}_1^\herm(\varphi, \theta) \vectr{e}_2(\varphi, \theta) = 0
$.
Thus, the radiation pattern of the two dual-polarized antennas remains identical to that of a single antenna element, while the radiated power is amplified by a factor of four, as $\norm{\be(\varphi, \theta)}^2 = 4$.

\section{Proposed BF Method for the Distributed Setup}

\begin{table}[tbp]
    \centering
    \normalsize
    \setlength{\tabcolsep}{5pt} 
    \renewcommand{\arraystretch}{1.2} 
    \small 
    \caption{The \gls{bf} coefficients for the \gls{dmimo} setup. \vspace{-5pt}}
    \label{tab:ap_slot_improved}
    \setlength{\tabcolsep}{5pt} 
    \begin{tabular}%
    {p{0.5cm}|p{2.8cm}|p{0.4cm}|p{2.8cm}}
        Slot & $AP_1$ & $\dots$ & $AP_L$ \\
        \hline
        $1$ &         $\bW_1^{(1)}$$=$$[\bw_{\text{V},1}^{(1)}\  \bw_{\text{H},1}^{(1)}]$ & 
        $\dots$ & 
        $\bW_L^{(1)} $$=$$ [\bw_{\text{V},L}^{(1)} \ \bw_{\text{H},L}^{(1)}]$ \\
        \hline
        $\vdots$ & $\vdots$ & $\vdots$ & $\vdots$ \\
        \hline
        $T$ & 
        $\bW_1^{(T)} $$=$$ [\bw_{\text{V},1}^{(T)} \ \bw_{\text{H},1}^{(T)}]$ & 
        $\dots$ & 
        $\bW_L^{(T)} $$=$$ [\bw_{\text{V},L}^{(T)} \ \bw_{\text{H},L}^{(T)} ]$ \\
        \hline
    \end{tabular}
    \vspace{-10pt}
\end{table}

In Table \ref{tab:ap_slot_improved}, the proposed \gls{dbf} coefficients for \gls{dmimo}, $\bW_l^{(t)} \!=\! [\bw_{\text{V},l}^{(t)} \ \bw_{\text{H},l}^{(t)}]$, are given. The vectors
$\bw_{\text{V},l}^{(t)}$ and $\bw_{\text{H},l}^{(t)}$ stand for the \gls{bf} coefficients for vertically and horizontally polarized antennas, respectively, in $\ap_l$ during time slot $t$, where $t \in \{1,\dotsc,T\}$. 
First, the \gls{bf} coefficients for each \gls{ap}, i.e., $\bW_l \!=\! [\bw_{\text{V},l} \ \bw_{\text{H},l}]$,  are designed based on the methods in \cite{petersson2020power, girnyk2020simple, petersson2022energy, girnyk2021efficient} to satisfy Eq. \eqref{eq:radiation_pattern}.
Then, the design of the coefficients for each time slot is calculated by 
\begin{equation}
    \bW_l^{(t)} = [\pilot]_{t,l} \bW_l,
\end{equation}
where $\pilot \in \complexset{T}{L}$ is a precoding matrix.

The total radiated energy from $L$ \glspl{ap} over $T$ time slots in the far-field region of the distributed \glspl{ap} is expressed as 
\begin{equation} 
    \sum_{t=1}^T \norm{ \sum_{l=1}^L \vectr{e}_l^{(t)}(\varphi, \theta)}^2 = \sum_\chi \sum_{t=1}^T \norm{\sum_{l=1}^L \ba^\trp_l(\varphi, \theta) \bw^{(t)}_{\chi,l}}^2,
\end{equation}
where $\vectr{e}_l^{(t)}(\varphi, \theta)$ is the radiated electric field from $\ap_l$ at time slot $t$, and $\ba(\varphi, \theta) = [\ba_1(\varphi, \theta), \dotsc, \ba_l(\varphi, \theta), \dotsc, \ba_L(\varphi, \theta)]^\trp$ is the steering vector size of $\sum_l K_l \times 1$. 
The vector $\ba_l(\varphi, \theta)$ is a part of the steering vector corresponding to $\ap_l$.
Without loss of generality, $G(\varphi, \theta) = 1$.

\begin{prop}
If $\pilot^\herm \pilot = \eye$, where $T \geq L$, the total radiated energy over $T$ time slots simplifies to
\vspace{-3pt}
\begin{equation}  \label{eq:radiated_energy}
    \sum_{t=1}^T \norm{ \sum_{l=1}^L \vectr{e}_l^{(t)}(\varphi, \theta)}^2 \stackrel{(a)}{=} \sum_\chi \sum_{l=1}^L \norm{\ba^\trp_l(\varphi, \theta) \bw_{\chi,l}}^2 = 2\sum_l K_l.
    \vspace{-3pt}
\end{equation}
This expression indicates that the radiated energy pattern mirrors the single-antenna case, scaled by a factor of $2\sum_l K_l$.
Another key insight is that the transmitted signals from different \glspl{ap} are orthogonal, as shown in $(a)$ in Eq. \eqref{eq:radiated_energy}, due to the proposed \gls{bf} design. This orthogonality holds for any arbitrary vector, not just $\ba_l(\varphi, \theta)$.

The matrix $\pilot$ can be chosen as a Hadamard matrix or a discrete Fourier transform (DFT) matrix, appropriately scaled to satisfy $\pilot^\herm \pilot = \eye$.
\end{prop}

\begin{proof}
For polarization $\chi$, the total radiated energy is 
\vspace{-5pt}
\begin{subequations} \label{eq:proof}
\begin{align}
    &\sum_{t=1}^T \norm{\sum_{l=1}^L [\pilot]_{t,l} 
 \underbrace{\ba^\trp_l(\varphi, \theta) \bw_{\chi,l}}_{c_l}}^2 \\
    &\quad = \norm{\pilot \bc}^2 = \bc^\herm \underbrace{\pilot^\herm \pilot}_{\eye} \bc = \sum_{l=1}^L \norm{\ba^\trp_l(\varphi, \theta) \bw_{\chi,l}}^2,
    \vspace{-5pt}
\end{align}
\end{subequations}
where $c_l$ is the $l$-th element of $\bc$.
Eq. \eqref{eq:proof} also holds when any arbitrary channel vector is used in place of 
$\ba^\trp_l(\varphi, \theta)$. Therefore, the transmitted signals from different \glspl{ap} are orthogonal to each other for any channel conditions.

In addition, since $\bW_l \!=\! [\bw_{\text{V},l} \ \bw_{\text{H},l}]$ satisfies Eq. \eqref{eq:radiation_pattern}, the total radiated energy across both polarizations can be expressed using Eq. \eqref{eq:proof} as
\vspace{-5pt}
\begin{equation}
    \sum_{l=1}^L \sum_\chi \norm{\ba^\trp_l(\varphi, \theta) \bw_{\chi,l}}^2 = 2\sum_l K_l.
    \vspace{-8pt}
\end{equation}
\vspace{-8pt}
\end{proof}

\begin{figure}[tbp]
	\centering
	\includegraphics[width = 0.6\linewidth]{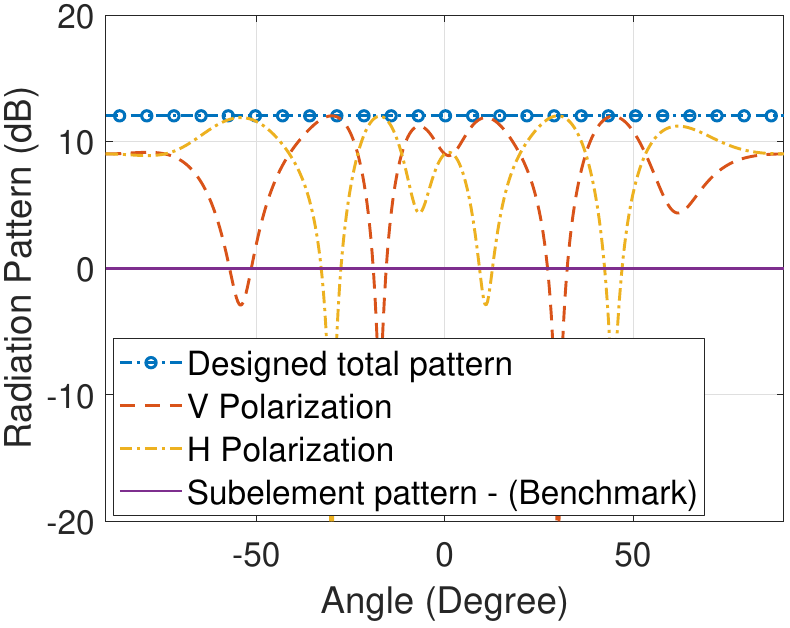}
    \vspace{-5pt}
	\caption{Radiation pattern for the dual-polarized array \cite{petersson2022energy}.}
	\label{fig:radiation_pattern}
    \vspace{-5pt}
\end{figure}

\section{Numerical Results} \label{sec:numerical_results}
For \gls{nlos} scenarios, we consider two reflectors as illustrated in Fig. \ref{fig:System_Model}. For \gls{los}, no reflectors are present, and the radiated wave travels in free space. The assumptions made in Section \ref{sec:e-field_direction} are also valid in this section.
The carrier frequency is $2.6$ GHz, and the inter-antenna distances are approximately $0.06$ \gls{m}. The coherence bandwidth of the system can be approximated as $B_\text{c} = c/d$ Hz, where $c$ is the speed of light and $d$ is the maximum difference between two propagation paths. For the given simulation setup, $d=10$ \gls{m} and $B_\text{c} = 30$ MHz.

Each antenna element in an \gls{ap} is assumed to have $G(\varphi, \theta)$$=$$1$. 
The service area of the \glspl{ap} has dimensions of
$10 \times 20$ \gls{m} in the $x$-$y$ plane.
We consider four cases:
\begin{itemize}
    \item Case 1 - \gls{cmimo} with \gls{dbf}: A single \gls{ap} with eight dual-polarized antennas is aligned along the $x$ axis and centered at $(x,y) = (5,10)$ m. The \gls{bf} coefficients $\bw_{\text{V},1}$ and $\bw_{\text{H},1}$ are designed as in \cite{girnyk2020simple,girnyk2021efficient}. This case serves as a benchmark for Case 2, i.e., \gls{cmimo} with \gls{ostbc}.
     \item Case 2 - \gls{cmimo} with \gls{ostbc}: The setup is the same as Case 1, but an \gls{ostbc} matrix $\bW \in \complexset{K_1}{K_1}$ is transmitted in both polarizations to broadcast information over $K_1$ time slots \cite{karlsson2018techniques}. The code matrix $\bW$ satisfies $\bW \bW^\herm  = \eye_{K_1}$.
    \item Case 3 - D-MIMO: Four \glspl{ap} $(L=4)$, each with two dual-polarized antennas $(K_l = 2)$ aligned along the $x$ axis, are placed at $(x,y) = (5,4)$, $(5,8)$, $(5,12)$, and $(5,16)$. 
    The \gls{bf} coefficients for the first \gls{ap} are designed using the method from Case 1. The remaining \glspl{ap} also use the same \gls{bf} coefficients with the first \gls{ap} as $\bw_{\text{V},l} = \bw_{\text{V},1}$ and $\bw_{\text{H},l} = \bw_{\text{H},1}$ for $l \in \{2,3,4\}$. This approach provides the simplest baseline in the distributed setup and serves as a benchmark for Case 4. 
    \item Case 4 - Proposed \gls{dmimo}: The setup is identical to Case 3, but the \gls{bf} coefficients are designed using the proposed method for the distributed setup, as shown in Table \ref{tab:ap_slot_improved}.
\end{itemize} 

In all cases, the total transmit energy remains the same. 
For each case, we calculate the \gls{pg} for \gls{nb} and \gls{wb} channels. The \gls{pg} is calculated in the $x$-$y$ plane at the height of the \glspl{ap}. The \gls{pg} in the case of \gls{nb} channel is calculated as 
\begin{equation} \label{eq:pg_calculation}
\begin{split}
    \text{PG} = 
    \begin{cases}
       \sum_\chi \norm{\sum_{n=0}^M {\vec{u}_{1,n}^\chi \bg_{\chi, 1}^{(n)}}^\trp \bw_{\chi,1} }^2 & \text{Case 1}, 
        \\
      \sum_\chi \norm{\sum_{n=0}^M {\vec{u}_{1,n}^\chi \bg_{\chi, 1}^{(n)}}^\trp \bW}^2 & \text{Case 2},
       \\
       \sum_\chi \norm{\sum_{l=1}^L \sum_{n=0}^M {\vec{u}_{l,n}^\chi \bg_{\chi, l}^{(n)}}^\trp \bw_{\chi,l}}^2 & \text{Case 3},
       \\
       \sum_\chi \sum_{t=1}^T \norm{\sum_{l=1}^L \sum_{n=0}^M {\vec{u}_{l,n}^\chi \bg_{\chi, l}^{(n)}}^\trp \bw_{\chi,l}^{(t)}}^2 & \text{Case 4},
    \end{cases}
    \end{split}
\end{equation}
where $M=2$ for \gls{nlos} channel, and  ${\bg_{\chi, l}^{(0)}}, {\bg_{\chi, l}^{(1)}}$, and 
${\bg_{\chi, l}^{(2)}}$ represent the \gls{los} channel and channels due to the reflection from first and second reflectors, respectively, and can be calculated using Eq. \eqref{eq:channel}. The channel between $\ap_l$ and the point where \gls{pg} is calculated is
${\bg_{\chi, l}} = {\bg_{\chi, l}^{(0)}} + {\bg_{\chi, l}^{(1)}} + {\bg_{\chi, l}^{(2)}}$, and it can be written as
$\bg_{\chi, l} = [g_{\chi, l, 1},\dotsc, g_{\chi, l, K_l}]^\trp$, where $g_{\chi, l, k_l}$ is given in Eq. \eqref{eq:channel}.
Note that the \gls{pg} is not normalized by the transmit power, so it also represents the received energy.

\begin{figure*}[tbp] 
    \centering 
    \subfloat[Case 3: \gls{dmimo}.]{
        \includegraphics[width=0.26\textwidth]{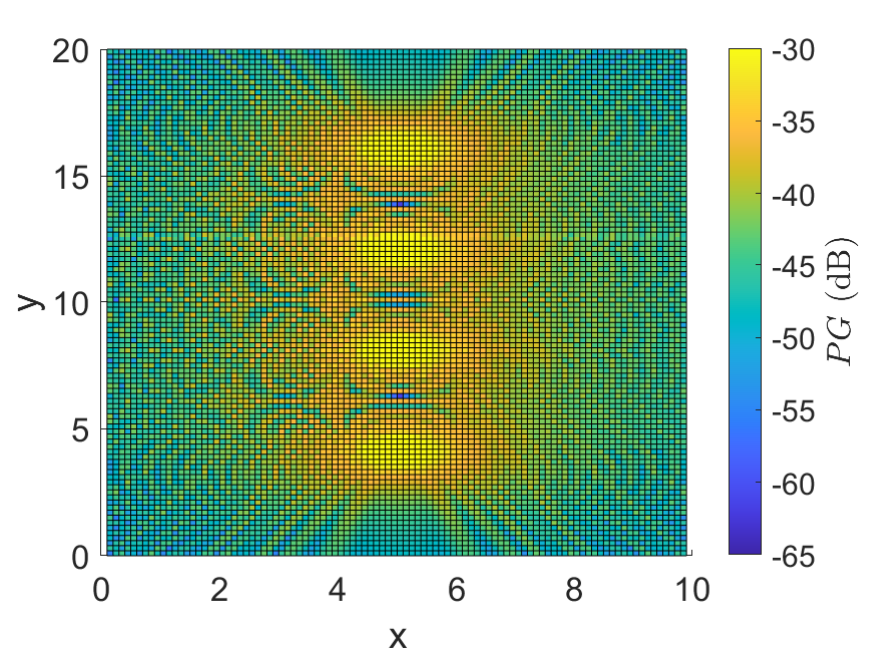}
        \label{fig:PG_distributed}
    }
    \hfil
    \subfloat[Case 4: Proposed \gls{dmimo}.]{
        \includegraphics[width=0.26\textwidth]{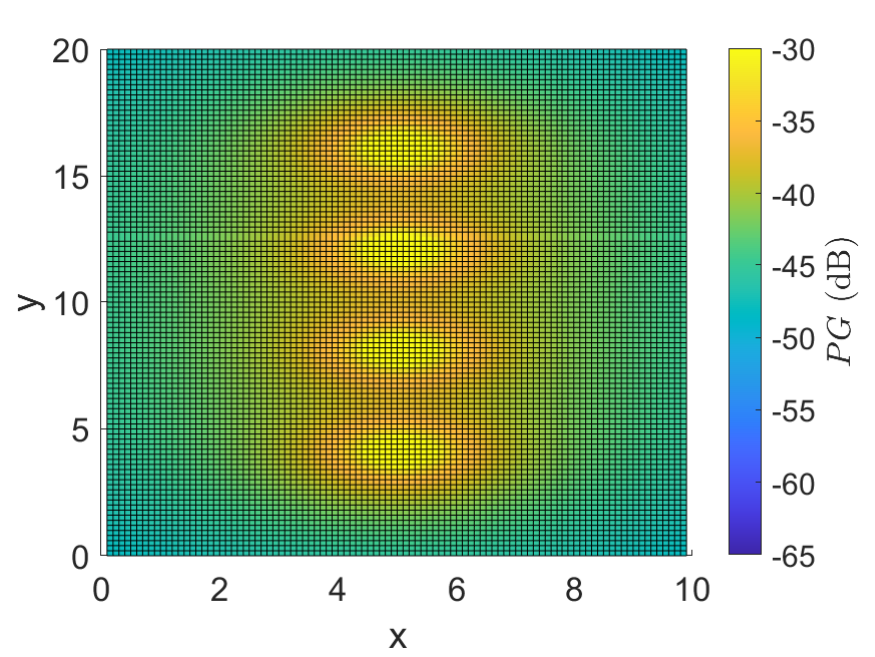}
        \label{fig:PG_distributed_proposed}
    }
    \hfil
    \subfloat[Empirical CDF in LoS channel.]{
    \includegraphics[width=0.24\textwidth]{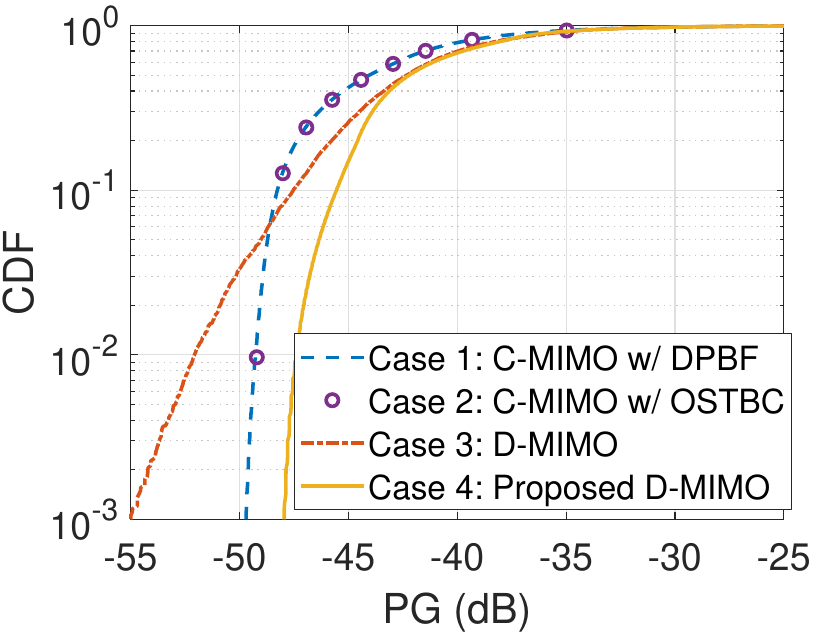}
        \label{fig:cdf_los}}
    \caption{Comparison of PG for collocated and distributed MIMO in LoS.}
    \label{fig:combined_figure}
    \vspace{-10pt}
\end{figure*}

\begin{figure*}[tbp]
    \centering
    \begin{minipage}{0.28\textwidth}
        \centering
        \includegraphics[width=\linewidth]{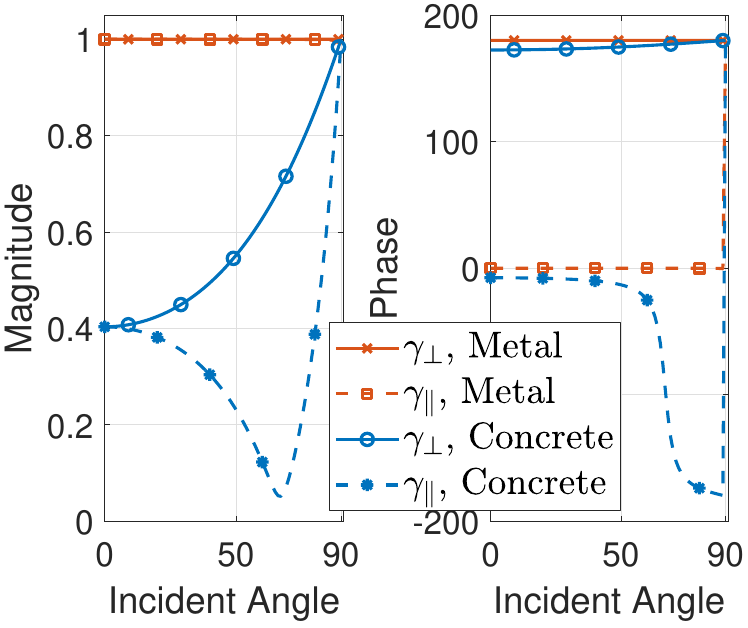}
        \caption{Reflection coefficients for different materials.}
        \label{fig:ref_coef}
    \end{minipage}
    \hfill
    \begin{minipage}{0.28\textwidth}
        \centering
        \includegraphics[width=\linewidth]{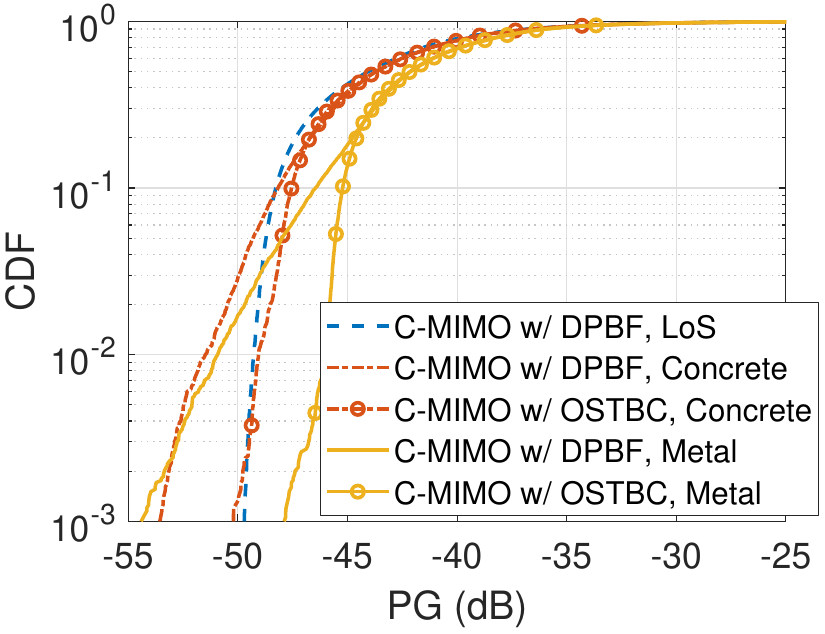}
        \caption{Empirical CDF in NLoS NB channel for C-MIMO.}
        \label{fig:cdf_nlos}
    \end{minipage}
    \hfill
    \begin{minipage}{0.28\textwidth}
        \centering
        \includegraphics[width=\linewidth]{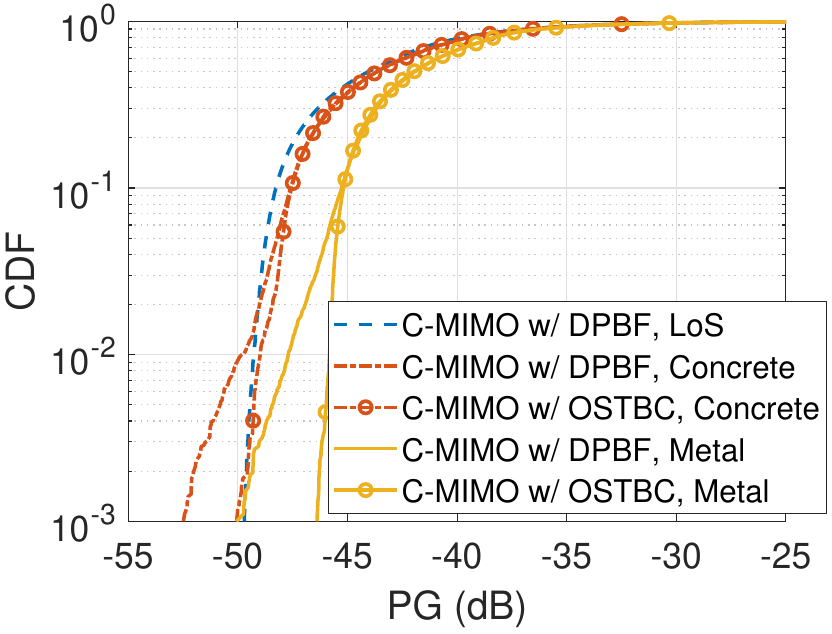}
        \caption{Empirical CDF in NLoS \gls{wb} channel for C-MIMO.}
        \label{fig:cdf_nlos_wb}
    \end{minipage}
    \vspace{-10pt}
\end{figure*}

\begin{table}[tbp]
    \centering
    \small
    \setlength{\tabcolsep}{6pt} 
    \renewcommand{\arraystretch}{1.2} 
    \caption{Material properties at $2.6$ GHz \cite{ITU-R_P.2040-3}\vspace{-5pt}}
    \label{tab:material_properties}
    \begin{tabular}{l c c c}
        \hline
        Property & Concrete Wall & Metal Reflector & Free Space \\
        \hline
        $\varepsilon$ [F/m] & $(5.2 - j0.6)\varepsilon_0$ & $(1 - j6 \times 10^7)\varepsilon_0$ & $\varepsilon_0$ \\
        $\mu$ [H/m] & $\mu_0$ & $\mu_0$ & $\mu_0$ \\
        $\sigma$ [S/m] & $0.1$ & $10^7$ & $0$ \\
        $\alpha$ [Np/m] & $8 + j0.5$ & $3.2 \times 10^{5}$ & $0$ \\
        $\beta$ [Rad/m] & $125 - j 8$ & $3.2 \times 10^{5}$ & $54.5$ \\
        $\eta$ [ohms] & $160 + j20$ & $0.02 + j0.02$ & $376.7$ \\
        \hline
    \end{tabular}
\end{table}

We also investigate the \gls{pg} distribution in the case of \gls{wb} channel, i.e., multiple subcarriers.
We use $100$ subcarriers separated uniformly in the frequency band of $100$ MHz centered at $2.6$ GHz. 
The channel for each subcarrier is calculated using Eq. \eqref{eq:channel} and the corresponding wavelengths for the subcarriers. 
We compute the \gls{pg} for each subcarrier using Eq. \eqref{eq:pg_calculation} and then calculate the average \gls{pg} over all subcarriers.
Note that we only calculate the reflection coefficients at $2.6$ GHz because they stay similar for the given frequency range.

\subsection{Results in \gls{los} Channel}
\subsubsection{\gls{los} Narrowband Channel}
Fig. \ref{fig:radiation_pattern} shows the \gls{ap} radiation pattern in the \gls{cmimo} with \gls{dbf} case. This figure is included for completeness and is a regenerated and updated version of \cite[Fig. 4]{petersson2022energy}. 
The radiation pattern is calculated using Eq. \eqref{eq:radiation_pattern}, where $G(\varphi, \theta) = 1$, $\varphi \in [-\pi, \pi)$, and $\theta = 0$.
As shown, the radiation pattern of the \gls{ap} remains identical to that of a single antenna element. Moreover, due to the phase-only \gls{bf}, all antennas operate at full power, resulting in a total radiated power that is $16$ times higher than that of a single antenna element. 

Figs. \ref{fig:PG_distributed} and \ref{fig:PG_distributed_proposed} illustrate the \gls{pg} distribution across the service area for Cases 3 and 4, respectively. 
As shown in figures, the transmitted signals from the four different \glspl{ap} in Case 4 remain orthogonal to each other, whereas in Case 3, destructive interference occurs.

In Fig. \ref{fig:cdf_los}, we compare the empirical \glspl{cdf} of \glspl{pg} for all cases in \gls{nb} \gls{los} channel. As shown in the figure, Case 4 achieves the best performance, benefiting from the proposed \gls{bf} coefficients and the macro-diversity effect of \gls{dmimo}. In contrast, there are deep fadings in Case 3
due to the destructive interference. In addition, \gls{cmimo} with \gls{ostbc} and \gls{dbf} performs the same in \gls{los} channel because both cases have the isotropic radiation pattern.

\subsubsection{\gls{los} Wideband Channel}
The results for \gls{los} \gls{wb} channel are not shown in any figure, as they are very similar to those in Fig. \ref{fig:cdf_los}, except for Case 3, whose performance closely aligns with Case 4.
This difference occurs because using subcarriers separated via large bandwidth changes the fading locations in the service area, and taking the average over subcarriers improves the performance due to the frequency diversity.  

\subsection{Results in \gls{nlos} Channel}

\subsubsection{Reflection Coefficient Analysis}
We use the metal reflectors and the concrete walls to analyze the system performance in \gls{nlos} channel.
The reflection coefficients are calculated using Eq. \eqref{eq:reflection_coeff}. The parameters for concrete wall, metal reflector, and free space at $2.6$ GHz are given in Table \ref{tab:material_properties}, where $\varepsilon_0 = 8.8\times 10^{-12}$ F/m and $\mu_0 = 4 \pi \times 10^{-7}$ H/m are the permittivity and permeability of a vacuum \cite{ITU-R_P.2040-3}. 

Fig. \ref{fig:ref_coef} shows the amplitudes and phases of reflection coefficients for the metal reflector and concrete wall at $2.6$ GHz.
As seen in Table \ref{tab:material_properties},  $\eta_0/\eta_m \gg 1$ is satisfied in the case of metal reflector. As a result,
$\gamma_{\perp,l,\medium} \approx -1$ and $\gamma_{\parallel,l,\medium} \approx 1$ for any angle of incidences except $90$ degrees as seen in Fig. \ref{fig:ref_coef}. The contribution of the specular multipath components is weaker in the case of concrete wall because the magnitude of the reflection coefficient is smaller.

\subsubsection{\gls{nlos} Narrowband Channel for C-MIMO}
Fig. \ref{fig:cdf_nlos} presents the empirical \gls{cdf} of the \gls{pg} for \gls{cmimo}. We compare the performances of Case 1 (\gls{cmimo} with \gls{dbf}) and Case 2 (\gls{cmimo} with \gls{ostbc})  under \gls{nlos} scenarios. 
The \gls{nlos} conditions are created using either two concrete walls or two metal reflectors located at $x=0$ and $x=10$.

From Fig. \ref{fig:cdf_nlos}, we can conclude that \glspl{smc} increase the variance of the \gls{pg} distribution for Case 1. 
For example, the performance of Case 1 with concrete is worse than Case 1 with \gls{los} with approximately $0.1$ probability.
More importantly, \gls{cmimo} with \gls{ostbc} outperforms \gls{cmimo} with \gls{dbf} across all scenarios in \gls{nlos} channel although they perform similarly in \gls{los} channel.
\footnote{Please note that, the average performance of the random \gls{bf} will also be the same with that of \gls{ostbc}. For random \gls{bf}, each \gls{bf} coefficients can be generated as $\exp(j2\pi \phi)$, with $\phi \sim \mathcal{U}[0, 2\pi)$ following a uniform distribution.}
However, note that \gls{cmimo} with \gls{ostbc} requires multiple time slots or subcarriers to achieve this performance. 

\subsubsection{NLoS Wideband Channel for C-MIMO} \label{sec:wideband}

Fig. \ref{fig:cdf_nlos_wb} presents a regenerated version of Fig. \ref{fig:cdf_nlos} for the \gls{wb} channel. 
As shown in Fig. \ref{fig:cdf_nlos_wb}, the performance gap in the \gls{pg} distributions of \gls{cmimo} with \gls{ostbc} and with \gls{dbf} decreases in the case of the \gls{wb} channel. However, \gls{cmimo} with \gls{ostbc} still performs best.

\subsubsection{NLoS Wideband and Narrowband Channel for D-MIMO} 

\begin{figure}[tbp]
	\centering
	\includegraphics[width = 0.6\linewidth]{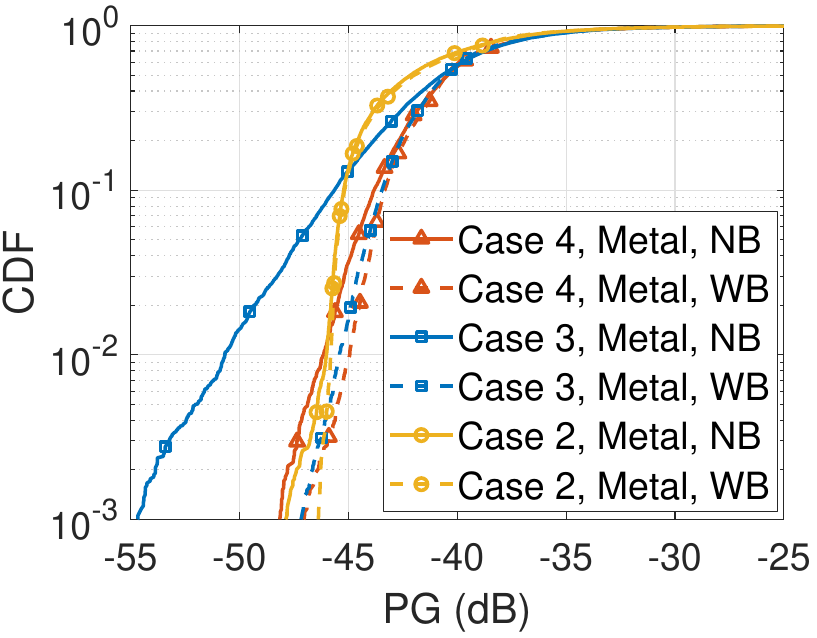}
	\caption{Empirical CDF in NLoS channel for D-MIMO.}
	\label{fig:cdf_dmimo_nlos}
    \vspace{-5pt}
\end{figure}

Fig. \ref{fig:cdf_dmimo_nlos} shows the \gls{pg} distribution for \gls{dmimo} in the \gls{nlos} channel for both \gls{nb} and \gls{wb}. As shown in the figure, Case 4 outperforms Case 3 in \gls{nb}, but the improvement is minimal in \gls{wb}. In addition, Case 4 achieves better \gls{pg} distribution compared to Case 2 (\gls{cmimo}) on average.

Table \ref{tab:results_summary} summarizes the results. Note that, for \gls{nlos} \gls{nb} channel, the relative performance of the cases may vary slightly depending on the carrier frequency since the fading points in the service area change. However, the results in Table \ref{tab:results_summary} remain unaffected as the \gls{pg} is calculated over a large area.

\begin{table}[tbp]
\centering
\caption{The summary of the results.\vspace{-5pt}}
\label{tab:results_summary}
\small
\begin{tabular}{|c|cc|}
\hline
\textbf{C-MIMO} & \multicolumn{1}{c|}{\textbf{Narrowband}} & \textbf{Wideband} \\ \hline
\textbf{LoS} & \multicolumn{2}{c|}{Cases 1 and 2 perform similarly.} \\ \hline
\textbf{NLoS} & \multicolumn{2}{c|}{\cellcolor[HTML]{C7F5C8}Case 2 (OSTBC) is better.} \\  \hline \hline
\textbf{D-MIMO} & \multicolumn{1}{c|}{\textbf{Narrowband}} & \textbf{Wideband} \\ \hline
\textbf{LoS} & \multicolumn{1}{c|}{\cellcolor[HTML]{C7F5C8}} &  \\ \cline{1-1}
\textbf{NLoS} & \multicolumn{1}{c|}{\multirow{-2}{*}{\cellcolor[HTML]{C7F5C8}\begin{tabular}[c]{@{}c@{}}Case 4 is better \\ than Case 3.\end{tabular}}} & \multirow{-2}{*}{\begin{tabular}[c]{@{}c@{}}Cases 4 and 3 \\ perform similarly.\end{tabular}} \\ \hline
\end{tabular}%
\end{table}

\section{Conclusion}
In this paper, we analyzed the performance of \gls{dbf} method for various \gls{mimo} configurations in both \gls{los} and \gls{nlos} channels. We modeled the \gls{nlos} channel considering the electric field directions and the reflection coefficients for different environmental materials. We show that \gls{cmimo} with \gls{ostbc} achieves better \gls{pg} distribution compared to the \gls{cmimo} with \gls{dbf} both in \gls{nb} and \gls{wb} \gls{nlos} channel, despite their similar performance in \gls{los} channel. However, \gls{ostbc} requires multiple time slots or subcarriers.
We also proposed the \gls{bf} method for broadcasting information in a \gls{dmimo} setup. We showed that the \gls{dmimo} with the proposed \gls{bf} coefficients is better than the benchmark scenario for \gls{dmimo} in \gls{nb} scenario and also achieves better \gls{pg} distribution compared to \gls{cmimo} due to the macro-diversity effect.


\linespread{0.98}\selectfont
\bibliographystyle{IEEEtran}
\bibliography{references}

\end{document}